\documentclass[runningheads]{llncs}
\pdfoutput=1

\usepackage[margin=1in]{geometry} 

\usepackage[T1]{fontenc}
\usepackage{graphicx}
\usepackage{tikz}
\usetikzlibrary{positioning,shapes,arrows,graphs}
\usepackage[numbers,sort&compress]{natbib}
\usepackage{comment}
\usepackage{amssymb}
\usepackage{amsmath}
\usepackage{xcolor}
\usepackage{thm-restate}
\usepackage{babel}

\makeatletter
\renewcommand\bibsection%
{
  \section*{\refname
    \@mkboth{\MakeUppercase{\refname}}{\MakeUppercase{\refname}}}
}
\makeatother

\usepackage[normalem]{ulem}

\newcommand{\todo}[1]{\textcolor{magenta}{#1}}
\DeclareMathOperator*{\argmax}{argmax}
\definecolor{gree}{rgb}{0, 0.6, 0}
\newcommand{\goods}{M} 
\newcommand{\alloc}{\boldsymbol{B}}

\newcommand{\MMS}{\mathsf{MMS}} %% For maximin welfare

\begin{document}
\title{Maximin Fair Allocation of Indivisible Items under Cost Utilities}
%
%\titlerunning{Abbreviated paper title}
% If the paper title is too long for the running head, you can set
% an abbreviated paper title here
%
\author{Sirin Botan\inst{1} \and
Angus Ritossa\inst{1} \and
Mashbat Suzuki\inst{1} \and
Toby Walsh\inst{1}}
\authorrunning{S. Botan et al.}
% First names are abbreviated in the running head.
% If there are more than two authors, 'et al.' is used.
%
\institute{UNSW Sydney\\ 
\email{\{s.botan, mashbat.suzuki, t.walsh\}@unsw.edu.au, a.ritossa@student.unsw.edu.au}}
\maketitle              % typeset the header of the contribution
\begin{abstract}
	
We study the problem of fairly allocating indivisible goods among a set of agents. Our focus is on the existence of allocations that give each agent their \emph{maximin fair share}---the value they are guaranteed if they divide the goods into as many bundles as there are agents, and receive their lowest valued bundle. An \emph{MMS allocation} is one where every agent receives at least their maximin fair share. We examine the existence of such allocations when agents have \emph{cost utilities}. In this setting, each item has an associated \emph{cost}, and an agent’s valuation for an item is the cost of the item if it is useful to them, and zero otherwise. 

 Our main results indicate that cost utilities are a promising restriction for achieving MMS. We show that for the case of three agents with cost utilities, an MMS allocation always exists. We also show that when preferences are restricted slightly further---to what we call \emph{laminar set approvals}---we can guarantee MMS allocations for any number of agents. Finally, we explore if it is possible to guarantee each agent their maximin fair share while using a strategyproof mechanism. 

\end{abstract}

\section{Introduction}

How to fairly divide a set of indivisible resources is a problem that has been studied by computer scientists, economists, and mathematicians \citep{M04,BCELP16,BT96}. Because of the fundamental nature of the problem, there is a large number of applications ranging from course allocations \citep{ASB10}, to division of assets \citep{GP15}, and air traffic management \citep{V02}. 

Among the fairness notions studied, two of the most commonly studied are those of envy-freeness---how to ensure no agent envies another, and maximin fair share---our focus in this paper. The notion of the maximin fair share  was introduced by \citet{B11}, and generalises the well known cut-and-choose protocol. Conceptually, an agent's \emph{maximin fair share} is the value they can achieve by partitioning the items into as many bundles as there are agents, and receiving their least preferred bundle. The ideal outcome is of course an \emph{MMS allocation}, where every agent receives at least their maximin fair share.

There has been a significant amount of work on MMS in the general additive valuations setting. \linebreak  Unfortunately, results are often quite negative. In general, MMS allocations cannot be guaranteed to exist, even in the case of three agents \citep{KPW18,FST21}. Furthermore, for instances where MMS allocations do exist (for example, when agents have identical valuations), computing an MMS allocation is NP-hard. As a result, a large body of work has been focused on establishing the existence of MMS allocations in more restricted settings \citep{AMNS17,BL16,bouveret2017fair,EPS22}. In this paper, we study the problem under a natural class of valuation functions--what we call \emph{cost utilities}---that allow us to provide fairness guarantees that are not achievable for general additive valuations. Cost utilities describe the setting where each item has an associated cost. An agent's value for any item is the cost of the item if it is useful to them, and zero otherwise. Our focus in this work is on the existence of MMS allocations under cost utilities. 

We are not the first to study this restriction in the context of fair division. \citet{BS06} provided an approximation of egalitarian welfare maximisation under cost utilities, that was then improved upon by \citet{AFS12} and \citet{CM18}. %On the other hand \cite{BD05} 
\citet{CFPT22} and \citet{ARS22} focus on envy-freeness, and show that an EFX allocation always exists under cost utilities\footnote{\citet{BS06} call them \emph{restricted assignment valuations}, while \citet{CFPT22} call them \emph{generalised binary valuations}.  \citet{ARS22} study them under the name \emph{restricted additive valuations}. We use the term ``cost utilities'' as we find it conceptually the most appealing and descriptive.}. There are clear practical advantages to studying this particular class of valuations. In many real-life settings, the price of items are known, and elicitation of preferences boils down to asking an agent whether they want the item %(at that price) 
or not---a task that can be accomplished easily.

\paragraph{Related Work.}

Given that MMS allocations cannot be guaranteed for general additive valuations, the work done on MMS in fair division has focused on two main approaches to circumvent this impossibility. The first---which is the route we employ in this paper---is to consider a restriction on the valuations of the agents. Examples of such restrictions under which MMS allocations always exist include binary valuations \citep{BL16}, and ternary valuations \citep{AMNS17}---where item values belong to $\{0,1,2\}$,  and Borda utilities \cite{HNNR18}. Existence of MMS allocations also holds for personalised bivalued valuations---where for each agent $i$, the value of an item belongs to $\{1,p_i\}$ for $p_i\in \mathbb{N}$, and weakly lexicographical valuations---where each agent values each good more than the combined value of all items that are strictly less preferred \citep{EPS22}.  %\todo{graph paper}%\mashsays{ I don't think graph paper fits here, because it is a constraint over possible allocations not on valuation functions}  

%A large body of work has been focused on establishing the existence MMS fair allocations both for general additive preferences, and in various restricted settings. These include the previously mentioned negative results in settings where agents have binary valuations \citep{BL16}, and three agents with additive utilities \citep{KPW18}. personalised bivalued valuations---where for each agent $i$, the value of each item belongs to $\{1,p_i\}$ where $p_i\in \mathbb{N}$---and weakly lexicographical valuations \citep{EPS22}.  

%\citet{KPW18} show, for example, that even in the case of three agents with additive utilities \citep{KPW18}, we cannot always find an MMS allocation.
%What about positive results 

%MMS allocations sometimes remain impossible even when agents have binary valuations \citep{BL16} \angussays{this isn't true, they show that MMS can always be found with binary valuations} ternary valuations (item values belong to $\{0,1,2\}$) \citep{AMNS17} \todo{mash how is this related to binary?} \angussays{the result for ternary in \citep{AMNS17} is positive, not negative}, 

%Furthermore, for instances when it does exist (such as identical valuations) computing an MMS allocation is NP-hard. 

The second approach is to examine \emph{how close} we can get to MMS, meaning how far each agent is from receiving their maximin fair share. %the computation and existence of $\rho$-MMS allocation. 
An allocation is said to be $\rho$-MMS, if each agent receives a $\rho$ fraction of their MMS value. 
\citet{GT20} show that for instances with more than five agents a $\left(\frac{3}{4}+\frac{1}{12n}\right)$-MMS allocation always exists. On the more negative side, \cite{FST21} show that there exist instances such that no allocation is $39/40$-MMS. For valuations that are beyond additive, the picture is arguably gloomier. \cite{GHSSY18} show an existence of $\frac{1}{3}$-MMS allocations and a PTAS for computing such allocations. They also show that for submodular valuations, there exist instances that do not admit any $\frac{3}{4}$-MMS allocation.   

There have been several works focused on achieving both fairness along with strategyproofness. \citet{ABCM17} show that when there are two agents and $m$ items there is no truthful mechanism that outputs an $\frac{1}{\lfloor m/2 \rfloor}$-MMS allocation. On the positive side \citet{HPPS20} and \citet{BEF21} show that when agents have binary valuations there is a polynomial time computable mechanism that is strategyproof and outputs an MMS allocation along with several other desirable properties.

%
%\begin{comment}
%In addition to maximin share, another well-studied fairness notions is that of \textit{envy-freeness} \cite{S80,BT95}. As envy-freeness cannot be guaranteed in the setting of indivisible items, the literature has turned to relaxations which are more suited. Two of the most common such  relaxations are that of \textit{envy-freeness up to one item} (EF1)\cite{B11} and  \textit{envy-freeness up to any item} (EFX)\cite{CKHPSW19}. It has been shown in \cite{LMMS04}, for monotone valuation functions an EF1 allocation always exists and can be computed in polynomial time. On the other hand only a few positive results on EFX is known, for instance two item types \cite{PKS22} and three agents \cite{CGM20}. Cost utilities have shown to be a useful restriction when studying EFX. It was shown  independently in \cite{CFPT22} and \cite{ARS22} that an EFX allocation exists under cost utilities\footnote{ With respect to the original definition of EFX suggested by \cite{CKHPSW19} also denoted EFX$_0$ or EFX$_{+}$}. 
%\end{comment} 

\paragraph{Our Contribution.}

We know that for some restricted settings---bivalued and ternary valuations---MMS allocations can always be found. \citet{ABFV22} highlight an open problem regarding the existence of other classes of structured valuations for which an MMS allocation is guaranteed to exist. Our paper answers this in the affirmative for a new class of valuation functions. We first show that MMS allocations exist for three agents under cost utilities, in contrast to the case of general additive utilities.  We also show that when valuations are restricted slightly further to laminar set approvals, MMS allocations are guaranteed to exist for any number of agents. Additionally, for the case of $n$ agents and $n+2$ items, we show there is a strategyproof polynomial time algorithm for computing Pareto optimal MMS allocations.

Interestingly, to the best of our knowledge, our results on cost utilities are first of its kind for which (other than identical valuations) the computation of the maximin fair share value is NP-hard, while existence of MMS allocation is still guaranteed. For previously known classes where an MMS allocation is guaranteed, the computation of the maximin fair share value can be done in polynomial time.% \todo{is this true for trees? and also we should put citations} \mashsays{ again, the tree paper is a restriction on the set of allocations not on the valuations so I don't think graph paper fits here}  

%\textcolor{red}{@Sirin I want to incorporate a problem raised by "Fair Division of Indivisible Goods: A Survey" by Amanatidis et al. They raise the following as an open problem which we do answer in a way.
%\textbf{Open Problem 6. Are there other classes of structured valuations for which MMS is guaranteed to exist, such as when there are only a few (but more than two) possible values?}   }
%Amanatidis et al. highlight an open problem regarding whether there are other classes of structured valuation functions for which an MMS allocation is guaranteed to exist. Our paper answers this in the affirmative for a large class of valuations. 
%
%We also hope the novel techniques used in our proofs are of interest to the community. lol

%Similarly, our results on cost utilities demonstrate the existence of MMS allocations (again, other than identical valuations) for a class of valuations where there are  $O(m)$ possible values an agent can have for an item. All previous results either had at most $O(n)$ (such as binary, personalised binary, ternary $\{0,1,2\}$) values or weakly lexicographical valuation where each agent values each goods more than combined value of all the items that are strictly less preferred. \mashsays{ I put it back as this was true because the tree paper is not a restriction on valuations}

%The results in this paper are also novel in the sense that they do not utilise the commonly used `valid reductions' technique. \todo{citations} 

\paragraph{Paper Outline.} 
In Section~\ref{sec:prelims} we introduce the framework of fair division of indivisible items, and present the central preference and fairness notions of the paper. Section~\ref{sec:mms} is focused on when we can achieve MMS allocations for cost utilities. Section~\ref{sec:strat} looks at a strategyproof mechanism for finding MMS allocations. Section~\ref{sec:concl} concludes.

\section{Preliminaries}\label{sec:prelims} 
Let $N$ be a set of $n$ \emph{agents}, and $\goods$ a set of $m$ indivisible \emph{goods} (or \emph{items}). Our goal is to divide $\goods$ among the agents in $N$ according to their preferences over the items. 

\paragraph{Preferences.} Each agent~$i \in N$ has a \emph{valuation function}~$v_i: 2^\goods \to \mathbb{R}_{\geq 0}$ that determines how much they value any bundle of items. For all agents~$i$, we assume that $v_i$ is additive, so $v_i(S) = \sum_{g \in S} v_i(g)$. For singleton bundles, we write $v_i(g)$ in place of $v_i(\{g\})$ for simplicity. %We assume that every item is approved by at least one agent---meaning there is no item $g \in \goods$ such that $v_i(g) = 0$ for all $i \in N$.\footnote{Our results also hold for cases where there are items that are unwanted by all agents. We can simply assign the item(s) to any agent.} 
We write $\boldsymbol{v} = (v_1, \dots, v_n)$ to denote the vector of all valuation functions for agents in $N$. 

Our focus in this paper is on a restricted domain of preferences---cost utilities. For these preferences, it is easy to think of each agent as submitting an approval set. Let $A_i$ be the \emph{approval set} of agent~$i$. More formally, we say $A_i = \{g \in \goods \mid v_i(g) > 0\}$. We say agents have \emph{cost utilities} if there exists a \emph{cost function} $c$ such that $v_i(S) = c(S \cap A_i)$ for all $S \subseteq \goods$ and all agents~$i \in N$. We require that the cost function is additive, as well as non-negative. 

\paragraph{Allocations and Mechanism.} An \emph{allocation}~$\alloc = (B_1, \dots, B_n)$ is an $n$-partition of the set of items $\goods$, where $B_i \subseteq \goods$  is the \emph{bundle} assigned to agent~$i$ under the allocation $\alloc$. We write ${\alloc}|_{N'}$ to denote the restriction of the allocation~$\alloc$ to only the bundles assigned to agents in $N' \subseteq N$. For a set of goods~$\goods$, we write $\mathcal{B}_n(\goods)$ to mean all possible allocations of the goods in $\goods$ to $n$ agents. %We require that $\bigcup_{i \in N} B_i = \goods$---all goods are allocated, and $B_i \cap B_j = \emptyset$ for all agents~$i,j \in N$---no good is allocated to more than one agent. 
An \emph{instance}~$\mathcal{I} = (N,\goods,\boldsymbol{v})$ of a fair division problem is defined by a set of agents, a set of goods, and the agents' valuations over those goods.

%
%We write $\mathcal{I} =_{-i} \mathcal{I'}$ if two instances $\mathcal{I}$ and $\mathcal{I'}$ differ only in terms of agent~$i$'s approval set. We write $\boldsymbol{\mathcal{I}}_n$ to denote the set of all instances with $n$ agents.
%
 Given an instance~$\mathcal{I}$, our goal is to find an allocation $\alloc$  that satisfies certain normative properties. An \emph{allocation mechanism} for $n$ agents and $m$ items is a function $f: \boldsymbol{V}_n \to \mathcal{B}_n(\goods),$
where $\boldsymbol{V}_n$ is the set of possible valuation profiles---i.e. vectors of $n$ valuation functions. 

\paragraph{Fairness and Efficiency.} For an agent~$i \in N$, their \emph{maximin fair share} in an instance~$\mathcal{I}= (N,\goods,\boldsymbol{v})$ is defined as %\todo{need to think about what we need as input here. at first the M seems redundant but then sometimes we do want to divide a smaller set of items even given the same instance. will get back to it after looking through everything else}

\[
\MMS^n_i(\mathcal{I}) = \max_{\alloc \in \mathcal{B}_n(\goods)} \min_{j \in N} v_i(B_j).
\] 

\noindent We sometimes write $\MMS^n_i(\goods)$ when the instance is clear from context. When the set of goods and the value of $n$ is fixed, we will also sometimes write $\MMS_i$. %Often the instance~$\mathcal{I}$, the number of agents~$n$, and/or the set of goods~$M$ is fixed. When this is the case, we sometimes omit reference to these variables and simply write $\MMS^n_i(\goods)$, $\MMS_i(\mathcal{I})$, or $\MMS_i$, whenever the context ensures this does not create confusion. \todo{clunky writing}

%Given a fair division instance $\mathcal{I} = (N,M,\boldsymbol{v})$, 
An \emph{MMS allocation}~$\alloc \in \mathcal{B}_n(\goods)$ is an allocation such that $v_i(B_i) \geq \MMS_i$ for all agents~$i \in N$. 

We say an allocation~$\alloc \in \mathcal{B}_n(\goods)$ is \emph{Pareto efficient} if there is no allocation~$\alloc' \in  \mathcal{B}_n(\goods)$ such that $v_i(B'_i) \geq v_i(B_i)$ for all $i \in N$ and $v_{i^*}(B'_{i^*}) > v_{i^*}(B_{i^*})$ for some ${i^*} \in N$. %\mashsays{Why bold $B_i$? }

\section{Maximin Fair Share Guarantees}\label{sec:mms} 

In this section, we will look at two settings where cost utilities can aid in finding cases where MMS allocations can be guaranteed to exist. Section~\ref{sec:three} focuses on cases with only three agents. Section~\ref{sec:laminar} considers any number of agents but is limited to laminar approval sets. This is a restriction that captures the idea of items belonging to different categories.

\subsection{MMS Allocations for Three Agents}\label{sec:three}
For the case of three agents, restricting our scope to considering only cost utilities yields positive results. As we have seen in the introduction, this is not the case for the more general case of additive preferences. Theorem~\ref{thm:3agents} is therefore a very welcome result. 

In this section, we will sometimes speak about items approved exclusively by two agents. We denote by $A_{ij} = (A_i \cap A_j) \setminus A_{i^*}$---where $i^* \in N$ and $i^* \neq i,j$---the set of items approved by agents~$i$ and $j$, and no third agent.

Before we state our main result in this section, we present the following two lemmas that we need in order to prove Theorem~\ref{thm:3agents}. Our first lemma simply tells us that adding items approved only by a single agent does not affect the existence of an MMS allocation. 

\begin{lemma}\label{lem:disjoint}
If an MMS allocation exists for instance~$\mathcal{I} = (N, M, \boldsymbol{v})$, then an MMS allocation also exists for the instance $\mathcal{I'} = (N, M \cup S, \boldsymbol{v})$, where $S$ is a set of items approved by a single agent~$i \in N$, and $S\cap M = \emptyset$. %when additional items approved only by one agent are added. %Given an agent~$i \in N$, and any two disjoint sets of items $S_1$ and $S_2$ such that $S_2$ is approved only by agent~$i$, if an MMS allocation exists for the instance $\mathcal{N, S_2, \boldsymbol{

%$v_i(T) +\MMS^n_i (S) \geq \MMS^n_i (S \cup T)$. %\ \ \  \forall i\in N, \forall k \in \mathbb{N}
\end{lemma}
\begin{proof}
Suppose we have an instance~$\mathcal{I} = (N, M, \boldsymbol{v})$ where $\alloc$ is an MMS allocation. Suppose further that $\mathcal{I'} = (N, M \cup S, \boldsymbol{v})$ is an instance where $S$ is a set of items approved by a single agent~$i \in N$, and $S\cap M = \emptyset$. We show that $\alloc'$ where $B'_j = B_j$ for all $j \neq i$ and $B'_i = B_i \cup S$ is an MMS allocation. Since for any $j\neq i$ we have  $v_j(B_j)\geq \MMS_j^n$, we only need to show that agent $i$ gets her MMS fair share.

Suppose for contradiction that we have $v_i(S) +\MMS^n_i (M) < \MMS^n_i (M \cup S)$. Let $\boldsymbol{W} = (W_1, \dots, W_n)$ be an $n$-partition of $(M \cup S)$ such that $v_i(W_k) \geq \MMS^n_i (M \cup S)$ for $1\leq k \leq n$. Note that for any $W_k$ in the partition we have that $W_k=(W_k\cap M)\cup (W_k\cap S)$. Thus we have the following: 
%Consider now the partition corresponding to $\MMS^k_i (S\cup T)$ denoted $(W_j)_{j\in [k]}$. Note that for each $j\in [k]$, we have 
%$$W_j=(W_j\cap S)\cup (W_j\cap T) $$ 
\begin{align*}
 \MMS^n_i (M\cup S)  & \leq v_i(W_k)  \\ 
& = v_i(W_k\cap M)+ v_i(W_k\cap S)  \\ 
& \leq v_i(W_k\cap M) + v_i(S) \\ 
& < v_i(W_k\cap M)+ \MMS^n_i (M\cup S) -\MMS^n_i(M)
\end{align*}
Where the last inequality follows from our assumption that $v_i(S) < \MMS^n_i (M \cup S) - \MMS^n_i (M)$. It follows that $v_i(W_k\cap M)> \MMS^n_i (M)$. As $k$ was chosen arbitrarily, this implies existence of a partition of $M$ into $n$ sets $(W_k\cap M)_{k\in [n]}$  such that each set has value strictly larger than $\MMS^n_i (M)$, a contradiction.\qed 
\end{proof}

Our second lemma is a more technical one. In yet another simplification of notation, we write $\mu_{ij} = \MMS^2_i(A_{ij}) = \MMS^2_j(A_{ij})$ to mean the maximin fair share of agents~$i$ and $j$ when dividing exactly the goods only the two of them approve among themselves. 

\begin{lemma}\label{lem:distinct}
Let $N = \{1,2,3\}$, and let $\boldsymbol{S} = (S_1, S_2, S_3)$ be a $3$-partition of $A_1$ such that $v_1(S_r) \geq \MMS_1$ for all $r \in \{1,2,3\}$. Then there exist distinct $k, \ell \in \{1,2,3\}$ such that 
\begin{equation*}
\begin{split}
c(S_k \cap A_{12}) &\leq \mu_{12}, \text{ and} \\
c(S_\ell \cap A_{13}) & \leq \mu_{13}.
\end{split}
\end{equation*}
\end{lemma} 

\begin{proof}
%More formally, we partition the set of goods~$M$ into three pairwise disjoint sets $S_1, S_2$, and $S_3$ such that $S_1 \cup S_2 \cup S_3 = \goods$ and $v_1(S_\ell) \geq \mu_1$ for all $\ell \in \{1,2,3\}$. Note that such a partition always exists by definition of $\mu_1$. 
%Let $A_{ij} = (A_i \cap A_j) \setminus A_{i^*}$ for $i,j \neq i^* \in N$---in other words, the items approved \emph{only} by agents~$i$ and $j$. 

Note that, by the definition of maximin fair share, there cannot be two elements $k_1, k_2 \in \{1,2,3\}$ such that $c(S_{k_1} \cap A_{12}) > \mu_{12}$ and $c(S_{k_2} \cap A_{12}) > \mu_{12}$---this would imply that we  could divide $A_{12}$ into two bundles such that both agents~$1$ and $2$ are guaranteed strictly more than their maximin fair share. 

Therefore, there must exist at least two distinct $k, k' \in \{1,2,3\}$ such that both $c(S_k \cap A_{12}) \leq \mu_{12}$ and $c(S_{k'} \cap A_{12}) \leq \mu_{12}$. The same argument tells us there are distinct $\ell, \ell' \in \{1,2,3\}$ such that $c(S_\ell \cap A_{13}) \leq \mu_{13}$ and $c(S_{\ell'} \cap A_{13}) \leq \mu_{13}$. Applying a pigeonhole argument, we conclude there must be distinct $k, \ell \in \{1,2,3\}$ such that $c(S_k \cap A_{12}) \leq \mu_{12}$ and $c(S_\ell \cap A_{13}) \leq \mu_{13}$, as desired.\qed 
\end{proof} 

We are now ready to state the main result of this section. 

\begin{theorem}\label{thm:3agents}
For three agents with cost utilities, there always exists a Pareto efficient MMS allocation.% always exists.
\end{theorem} 

\begin{proof}
Given a set of agents $N=\{1,2,3\}$, let $\MMS_i = \MMS_i^3(\goods)$---the maximin fair share of agent~$i$ when dividing the items in $\goods$ among the three agents. We assume that for any item $g \in \goods$, we have that $g$ is approved by at least two agents. By Lemma~\ref{lem:disjoint}, we know the claim will also hold for the remaining cases where there are additional goods approved by a single agent. 

Finally, we define the following three values: 
\begin{equation*}
\begin{split}
& q_1 = \MMS_1 + \mu_{23} \\ 
& q_2 = \MMS_2 + \mu_{13} \\ 
& q_3 = \MMS_3 + \mu_{12}
\end{split} 
\end{equation*}

\noindent Without loss of generality, we assume that $q_1 \geq q_2$ and $q_1 \geq q_3$. We can rewrite this, and express it as follows: 
\begin{equation}\label{ineq:1}
\MMS_1 + \mu_{23} - \mu_{13} \geq \MMS_2
\end{equation}
\begin{equation}\label{ineq:2}
\MMS_1 + \mu_{23}  - \mu_{12} \geq \MMS_3
\end{equation}

\noindent Our method for finding an allocation that satisfies the maximin property and is Pareto efficient, takes as basis a partition of the goods where each bundle reaches the maximin fair share of agent~$1$. Let $\boldsymbol{S} = (S_1, S_2, S_3)$ be a $3$-partition of $A_1$ such that $v_1(S_r) \geq \MMS_1$ for all $r \in \{1,2,3\}$. %---in other words, a partition that might result from asking agent~$1$ to divide the goods into three bundles while ensuring they are all valued at least his minimax fair share. 
Note that such a partition always exists by definition of $\MMS_1$. By Lemma~\ref{lem:distinct} we know there exist distinct $k, \ell \in \{1,2,3\}$ such that
\begin{equation}\label{ineq:3}
c(S_k \cap A_{12}) \leq \mu_{12}, 
\end{equation} 
\begin{equation}\label{ineq:4}
c(S_\ell \cap A_{13}) \leq \mu_{13}.
\end{equation} 

\noindent We can now describe the allocation~$\alloc$, which we claim is a Pareto efficient MMS allocation. 

We divide $A_{23}$ into two disjoint sets $T_1$ and $T_2$ such that $c(T_1) \geq \mu_{23}$ and $c(T_2) \geq \mu_{23}$. Note that such a partition exists by the definition of $\mu_{23}$. Let $S_x$ be the third bundle in $\boldsymbol{S}$---i.e. $x \in \{1,2,3\} \setminus \{k,\ell\}$. We then allocate the goods in $\goods$ as follows:

\begin{equation*}
\begin{split}
B_1 &= (S_\ell\setminus{A_2}) \cup (S_k\setminus{A_3}) \cup S_{x}\\
B_2 & = (S_\ell \cap A_2) \cup T_1\\
B_3 &= (S_k \cap A_3) \cup T_2
\end{split}
\end{equation*}

\noindent In words, agent~$2$ receives $T_1$ and everything in $S_\ell$ that she wants, agent~$3$ receives $T_2$ and everything in $S_k$ that she wants, and agent $1$ receives the remaining items in $S_k$ and $S_\ell$ as well as the entire bundle $S_x$. Note that all items have been allocated as $A_1 \cup A_{23} = M$, and no item is allocated to more than one agent as $S_x$ and $A_{23}$ are disjoint. By definition, we have that $v_1(B_1) \geq \MMS_1$---agent 1 clearly receives their maximin fair share as she receives one of the original bundles, $S_x$, and then some. We now show that the same must hold for the other two agents. 

For agent~$2$, we need to show that $v_2(B_2) \geq \MMS_2$. Note that we can express the value of agent~2's bundle using the cost function~$c$ as follows (where $S_\ell \cap A_{13}$ is the portion of $S_\ell$ that agent~$2$ values at $0$).\footnote{This is possible because we know that any good in the set is either approved by all three agents, or a subset of two. Agent 2 is a member of any subset of size two except $A_{13}$.}
\begin{align*}
v_2(B_2) 
&=  v_2(S_\ell \cap A_2) + v_2(T_1) \\
&= c(S_\ell \cap A_2) + c(T_1)\\
&= c(S_\ell) - c(S_\ell \cap A_{13}) + c(T_1) 
\end{align*}

\noindent Because of the way we've defined the partition~$\boldsymbol{S}$ and $A_{13}$, we know that $c(S_\ell) \geq \MMS_1$ and $c(T_1) \geq \mu_{23}$. Additionally, by Equation~\ref{ineq:4}, we know that $c(S_\ell \cap A_{13}) \leq \mu_{13}$. From this, we can conclude the following, where the last inequality follows from Equation~\ref{ineq:1}. %\todo{here we need to be clear about what happens to the stuff only a single agent likes}
\begin{align*}
v_2(B_2) 
&= c(S_\ell) - c(S_\ell \cap A_{13}) + c(T_1) \\
&\geq \MMS_1 - \mu_{13} + \mu_{23} \\
& \geq \MMS_2
\end{align*}
 
\noindent Putting this all together, we have shown that $v_2(B_2) \geq \MMS_2$, as desired. The proof for agent~$3$ proceeds analogously, using Equations~\ref{ineq:2} and~\ref{ineq:3}. Thus, we have shown that $\alloc$ is an MMS allocation.

Finally, we see that no item has been allocated to an agent who values it at $0$, meaning the allocation is indeed Pareto efficient.\qed   
\end{proof}

Theorem~\ref{thm:3agents} establishes a clear improvement when dealing with cost utilities over general additive valuations. 

\subsection{MMS Allocations for Laminar Set Approvals}\label{sec:laminar}

In this section we present our results for agents with laminar set approvals. This restriction on the agents' preferences has a very natural interpretation, in that it describes the notion of items falling into categories and subcategories quite well. %Combined with cost utilities it also captures the case where agents might have a price cap for the items they like. \todo{phrase better}
We can think of agents as approving categories as a whole. For example, one agent might want all vegetarian dishes, while another wants only the seafood. A third agent might want the pasta-based vegetarian dishes, which would constitute a subcategory of vegetarian. 

%It is also well studied within economics, computer science and mathematics. \todo{cant claim this without citation} 

We say agents with cost utilities have \emph{laminar set approvals} if for a vector~$\boldsymbol{A} = (A_1, \dots, A_n)$ of approval sets, we have that for any $i, j \in N$, either $A_i \cap A_j = A_j$,  $A_i \cap A_j = \emptyset$, or $A_i \cap A_j = A_i$. In words, for any two agents, one approval set is either a subset of the other, or the sets are disjoint. Note that in this paper, we only examine laminar set approvals within the context of cost utilities.

We first present a technical lemma that we will apply inductively in the proof of Theorem~\ref{thm:laminar}. Lemma~\ref{lem:lam-ind} allows us to carry the existence of an MMS allocation from cases where all agents submit the whole set~$M$ of goods as their approval, to cases where fewer and fewer agents do so, until we reach a single agent approving all goods. 

\begin{lemma}\label{lem:lam-ind}
For $n$ agents with cost utilities and laminar set approvals, and $k\geq 1$, if an MMS allocation exists for all instances where $k+1$ agents approve all items in $\goods$, then an MMS allocation exists for any instance where $k$ agents approve all items.
\end{lemma}
\begin{proof}
Consider an instance $\mathcal{I}=(N,\goods,\boldsymbol{v})$ where there are $k \geq1$ agents whose approval set equals $\goods$. We call this set of agents $N'$.  Let $i \in N \setminus N'$ be an agent such that $A_i \not\subset A_j$ for all $j \in N \setminus N'$ in the instance~$\mathcal{I}$. Note that such an agent must exist, as agents have laminar set approvals. See Figure~\ref{fig:laminarlemma} for a visual representation. We will continue to use this figure throughout this proof. %\todo{Then for any $j \in N\setminus N'$ it must either be that $A_i = A_j$ or $A_j \subset A_i$. }

\begin{figure}
	\begin{center}
	\scalebox{0.75}{
			\begin{tikzpicture}

\fill [fill=gray, fill opacity=0.25] (0,0) ellipse (2.7cm and 2cm); %M
\fill [rotate = 0, fill=gray, fill opacity=0.25] (-1.1,0) ellipse (0.9cm and 1.3cm); 
\fill [rotate = 0, fill=gray, fill opacity=0.25] (1.1,0) ellipse (0.9cm and 1.3cm);
\fill [rotate = 40, fill=gray, fill opacity=0.25] (1,-0.3) ellipse (0.7cm and 0.4cm);
\fill [rotate = 40, fill=gray, fill opacity=0.25] (0.6,-1.2) ellipse (0.7cm and 0.4cm);
\fill [rotate = 90, fill=gray, fill opacity=0.25] (0,1.1) ellipse (0.6cm and 0.5cm); 

\fill [rotate = 70, fill=blue, fill opacity=0.25] (0.5,1.1) ellipse (0.8cm and 0.6cm); 

\node at (0, 1.4) {$B'_{i^*}$}; %bag
%\node at (-0.9, 0.3) {$c$}; %coat
%\node at (0.8, 0.3) {$s_1$}; %shoe 1
%\node at (0.6, -0.3) {$t_1$}; %trouser 1
%\node at (1, -0.3) {$t_2$}; %trouser 2
%\node at (0, 0.8) {$s_2$}; %shoe2

\node at (0, -2.4) {$\goods = A_1, \dots, \underline{A_{i^*}}, \dots, A_k, \underline{A'_i}$};
\node at (-1.1, -1.5) {$A_i$}; %A_i
%\node at (1.57, -0.7) {$B$};

\end{tikzpicture}
}
\end{center}
\caption{An illustration of the sets involved in the proof of Lemma~\ref{lem:lam-ind}. The largest set is $\goods$---the set of goods. Note how the approval sets $A_1, \dots, A_k$ are equivalent to the whole set of goods, and the same holds for $A_{i^*}$ and $A'_i$. The approval set $A_i$ is ``one level below'' the sets approving all items. The bundle $B'_{i^*}$---represented in blue---is the bundle in the allocation~$\alloc|_{N' \cup \{i\}}$ that is highest valued according to $v_i$.}\label{fig:laminarlemma}
\end{figure}
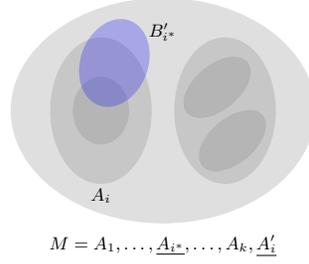
Our aim is to show that there exists an MMS allocation for the instance~$\mathcal{I}$. To this end, we define a second instance $\mathcal{I'} =(N,\goods,\boldsymbol{v'})$ such that $A'_i = \goods$, and $A'_j = A_j$ for all agents $j \neq i$---i.e. the instance~$\mathcal{I'}$ only differs from $\mathcal{I}$ in that agent~$i$ now approves all items. Thus, we have $k+1$ agents whose approval set is $\goods$ in the instance $\mathcal{I'}$. 

Suppose $\alloc'$ is an MMS allocation for $\mathcal{I'}$, such an allocation is guaranteed to exist by the assumption of the lemma. We construct an MMS allocation~$\alloc$ for our initial instance by building on $\alloc'$. We first define $i^* \in \argmax_{j \in N' \cup \{ i \}} v_i(B'_j)$. This is an agent who gets the highest value bundle in $\alloc|_{N' \cup \{i\}}$ according to $v_i$---agent~$i$'s valuation in the initial instance. Because the value $n$ is fixed, we will write $\MMS_i(\mathcal{I})$ to mean $\MMS^n_i(\mathcal{I})$. We consider two cases. 

\textit{Case 1:} Suppose $v_i(B'_{i^*}) \geq \MMS_i(\mathcal{I})$. Then agent~$i$ values agent~$i^*$'s bundle at least as much as their maximin fair share in the initial instance. We define an allocation $\alloc$ and claim that it is an MMS allocation for the instance $\mathcal{I}$. 
\[
B_j = \begin{cases} B'_{i^*}  \ \ \ \ &\text{if} \ \ \  j= i \\ 
B'_{i}  \ \ \ \ &\text{if} \ \ \  j = i^* \\ 
B'_{j} \ \ \ \ &\text{otherwise} %\ \ \ j \neq i \ \ \ \text{and} \ \ \ j \neq i^* \\ 
\end{cases}
\]

\noindent First note that for any agent $j \not\in \{i, i^* \}$, their maximin fair share is the same across both instances, and they receive the same bundle under $\alloc$ and $\alloc'$. Thus, they receive at least their maximin fair share in the allocation $\boldsymbol{B}$. 

We now show the same holds for $i$ and $i^*$.  %Thus it only remains to show that agents $i$ and $i^*$ get their MMS fair share under the allocation $\alloc$.
For agent~$i$, 
%If $i = i^*$, then 
this follows by assumption since $v_i(B_i) = v_i(B'_{i^*}) \geq \MMS_i(\mathcal{I})$. 
%On the other hand, if $i \neq i^*$, then agent $i$ receives their maximin fair share by assumption. \todo{also ?}
For agent~$i^*$ then, we only need to consider when $i^* \neq i$. In that case, as $i^* \in N'$, we have that 
%Additionally, since 
$A_{i^*} =  A'_{i} = M$. Then agent $i^*$ must also receive their maximin fair share in the allocation $\boldsymbol{B}$, because $v_{i^*}(B_{i^*}) = v'_{i}(B'_{i}) \geq \MMS_i(\mathcal{I'}) = \MMS_{i^*}(\mathcal{I})$. Note that this holds because the agents have cost utilities, and both $v_{i^*}$ and $v'_i$ are equivalent to the cost function~$c$ since $A_{i^*} = A'_i = M$. As $\boldsymbol{B}$ guarantees everyone at least their maximin fair share, it is an MMS allocation for ~$\mathcal{I}$.

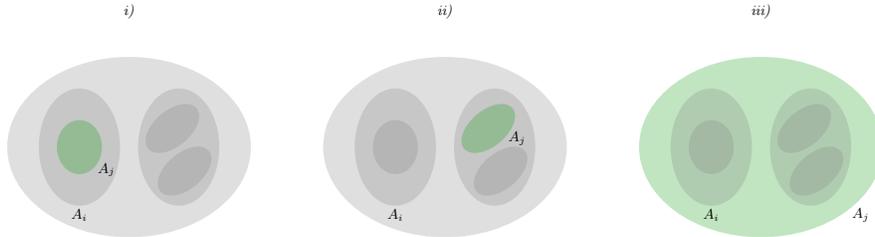
\begin{figure}
	\begin{center}
	\scalebox{0.6}{
	\begin{tikzpicture}
	
\node at (-7,0) {
			\begin{tikzpicture}

\fill [fill=gray, fill opacity=0.25] (0,0) ellipse (2.7cm and 2cm); %M
\fill [rotate = 0, fill=gray, fill opacity=0.25] (-1.1,0) ellipse (0.9cm and 1.3cm); 
\fill [rotate = 0, fill=gray, fill opacity=0.25] (1.1,0) ellipse (0.9cm and 1.3cm);
\fill [rotate = 40, fill=gray, fill opacity=0.25] (1,-0.3) ellipse (0.7cm and 0.4cm);
\fill [rotate = 40, fill=gray, fill opacity=0.25] (0.6,-1.2) ellipse (0.7cm and 0.4cm);
\fill [rotate = 90, fill=gree, fill opacity=0.25] (0,1.1) ellipse (0.6cm and 0.5cm); 

%\fill [rotate = 0, fill=red, fill opacity=0.25] (-1.1,0.2) ellipse (0.6cm and 0.7cm); 

\node at (-0.5, -0.5) {$A_j$}; %bag
%\node at (-0.9, 0.3) {$c$}; %coat
%\node at (0.8, 0.3) {$s_1$}; %shoe 1
%\node at (0.6, -0.3) {$t_1$}; %trouser 1
%\node at (1, -0.3) {$t_2$}; %trouser 2
%\node at (0, 0.8) {$s_2$}; %shoe2

%\node at (0, -2.4) {$\goods$};
\node at (-1.1, -1.5) {$A_i$}; %A_i
%\node at (1.57, -0.7) {$B$};

\end{tikzpicture}};

\node at (0,0) {
			\begin{tikzpicture}

\fill [fill=gray, fill opacity=0.25] (0,0) ellipse (2.7cm and 2cm); %M
\fill [rotate = 0, fill=gray, fill opacity=0.25] (-1.1,0) ellipse (0.9cm and 1.3cm); 
\fill [rotate = 0, fill=gray, fill opacity=0.25] (1.1,0) ellipse (0.9cm and 1.3cm);
\fill [rotate = 40, fill=gree, fill opacity=0.25] (1,-0.3) ellipse (0.7cm and 0.4cm);
\fill [rotate = 40, fill=gray, fill opacity=0.25] (0.6,-1.2) ellipse (0.7cm and 0.4cm);
\fill [rotate = 90, fill=gray, fill opacity=0.25] (0,1.1) ellipse (0.6cm and 0.5cm); 

%\fill [rotate = 70, fill=blue, fill opacity=0.25] (0.5,1.1) ellipse (0.8cm and 0.6cm); 

%\node at (0, 1.4) {$B'_{i^*}$}; %bag
%\node at (-0.9, 0.3) {$c$}; %coat
%\node at (0.8, 0.3) {$s_1$}; %shoe 1
%\node at (0.6, -0.3) {$t_1$}; %trouser 1
%\node at (1, -0.3) {$t_2$}; %trouser 2
%\node at (0, 0.8) {$s_2$}; %shoe2

\node at (1.6, 0.2) {$A_j$}; %bag

%\node at (0, -2.4) {$\goods$};
\node at (-1.1, -1.5) {$A_i$}; %A_i
%\node at (1.57, -0.7) {$B$};

\end{tikzpicture}};

\node at (7,0) {
			\begin{tikzpicture}

\fill [fill=gree, fill opacity=0.25] (0,0) ellipse (2.7cm and 2cm); %M
\fill [rotate = 0, fill=gray, fill opacity=0.25] (-1.1,0) ellipse (0.9cm and 1.3cm); 
\fill [rotate = 0, fill=gray, fill opacity=0.25] (1.1,0) ellipse (0.9cm and 1.3cm);
\fill [rotate = 40, fill=gray, fill opacity=0.25] (1,-0.3) ellipse (0.7cm and 0.4cm);
\fill [rotate = 40, fill=gray, fill opacity=0.25] (0.6,-1.2) ellipse (0.7cm and 0.4cm);
\fill [rotate = 90, fill=gray, fill opacity=0.25] (0,1.1) ellipse (0.6cm and 0.5cm); 

%\fill [rotate = 0, fill=red, fill opacity=0.25] (0.5,1.1) ellipse (0.8cm and 0.6cm); 

%\node at (0, 1.4) {$B'_{i^*}$}; %bag
%\node at (-0.9, 0.3) {$c$}; %coat
%\node at (0.8, 0.3) {$s_1$}; %shoe 1
%\node at (0.6, -0.3) {$t_1$}; %trouser 1
%\node at (1, -0.3) {$t_2$}; %trouser 2
%\node at (0, 0.8) {$s_2$}; %shoe2

\node at (2.21, -1.5) {$A_j$};
%\node at (0, -2.4) {$\goods$};
\node at (-1.1, -1.5) {$A_i$}; %A_i
%\node at (1.57, -0.7) {$B$};

\end{tikzpicture}};

\node at (-7,3) {\textit{i)}};
\node at (0,3) {\textit{ii)}};
\node at (7,3) {\textit{iii)}};
\end{tikzpicture}
}
\end{center}
\caption{An illustration of the sets involved in Case 2 of the proof of Lemma~\ref{lem:lam-ind}. The possible approval set of agent~$j$ in each case is represented in green.}\label{fig:laminarlemma2}
\end{figure}

\textit{Case 2:} Suppose instead that $v_i(B'_{i^*}) < \MMS_i(\mathcal{I})$. In this case, agent~$i$ values agent~$i^*$'s bundle strictly less than their maximin fair share in the initial instance. Recall that $v_i(B'_j) \leq v_i(B'_{i^*})$ for all $j \in N' \cup \{ i \}$---agent $i^*$'s bundle is still the ``best'' one among those in~$\alloc|_{N' \cup \{i\}}$. Given our initial assumption, we then have that
\begin{equation}\label{eq:lem3}
v_i(B'_j) < \MMS_i(\mathcal{I}) \text{ for all } j \in N' \cup \{ i \}.
\end{equation}

\noindent Before we proceed, we will need to define a third instance over only the goods in $A_i$. Let $\mathcal{I^*} = (N,A_i,\boldsymbol{v})$ be a restriction of the instance~$\mathcal{I}$ to only the items in $A_i$---meaning $A^*_j = A_j \cap A_i$ for all $j \in N$. Note that in $\mathcal{I^*}$, there are at least $k+1$ agents whose approval set is $A_i$---the initial $k$ agents who approved all items in~$\mathcal{I}$, and agent~$i$. Let $\alloc''$ be an MMS allocation for $\mathcal{I^*}$. We now proceed with defining an allocation $\alloc$ by using both allocations~$\alloc'$ and $\alloc''$. In particular, we define
$$B_j = (B'_j \setminus A_i) \cup B''_j \text { for all } j \in N.$$ 

\noindent Note that no item is allocated more than once because $B''_j \subseteq A_i$ for all $j \in N$. We claim that $\alloc$ is an MMS allocation for the instance $\mathcal{I}$. Because agents have laminar set approvals, there are three possible cases for any agent~$j$: either \textit{i)}  $A_j \subseteq A_i$, or \textit{ii)} $A_j \cap A_i = \emptyset$, or \textit{iii)} $A_i \subset A_j$. See Figure~\ref{fig:laminarlemma2} for a visual representation.%We now show that $\alloc$ gives any agent $j \in N$ their maximin fair share. %Take any agent $j \in N$. 

\begin{itemize}
    \item[\textit{i)}] Suppose $A_j \subseteq A_i$.  Then agent~$j$ was only approving items in $A_i$ and their approval set remains the same in the restriction~$\mathcal{I}^*$, implying that their maximin fair share also remains the same in both instances. Additionally, we have that $v_j(B'_j \setminus A_i) = 0$ given that $A_j \subseteq A_i$, and so $v_j(B_j) = v_j(B''_j)$. Since $j$ receives their maximin fair share in $\alloc''$, they also do so in $\alloc$.

    \item[\textit{ii)}] Suppose instead $A_j \cap A_i = \emptyset$.  Because agent~$j$ does not approve any items in $A_i$, we have that $v_j(B'_j) = v_j(B'_j \setminus A_i)$ and $v_j(B''_j) = 0$. %, since $v_j(A_i) = 0$. 
Then $v_j(B_j) = v_j(B'_j)$, and because $A'_j = A_j$ their maximin fair share is the same in $\mathcal{I}$ and $\mathcal{I'}$. 
Thus $j$ receives their maximin fair share in $\alloc$. %since $v_j(B_j) \geq v_j(B'_j) \geq \MMS_j(\mathcal{I'}) = \MMS_j(\mathcal{I})$.
    
   \item[\textit{iii)}] Finally, suppose $A_i \subset A_j$. This is only possible if $j \in N'$, meaning $j$ is one of the agents approving all items. We know that 
   \vspace{3pt}
    \begin{equation}\label{eq:lam}
    \begin{split}
	v_j(B_j) & = v_j(B'_j \setminus A_i)  + v_j(B''_j) \\
	& = v_j(B'_j) - v_j(B'_j \cap A_i)  + v_j(B''_j) \\
	& =  v_j(B'_j)  - v_i(B'_j) + v_j(B''_j)
	\end{split}
    \end{equation} 
\noindent Where the last line follows from the fact that agents have cost utilities, meaning $v_j(B'_j \cap A_i) = v_i(B'_j)$.

Recall that Equation~\ref{eq:lem3} tells us $v_i(B'_j) < \MMS_i(\mathcal{I})$. This fact, combined with Equation~\ref{eq:lam} (and some reshuffling of the terms), tells us it must be the case that % < \MMS_i(\mathcal{I})$, it follows that 
     \begin{equation}\label{eq:lam2} 
     v_j(B_j) > v_j(B'_j) - \MMS_i(\mathcal{I}) + v_j(B''_j).
     \end{equation}   
    
    Since $\alloc''$ is an MMS allocation for $\mathcal{I^*}$, it follows that $v_j(B''_j) \geq \MMS_j(\mathcal{I^*})$. Further,  since $A_i \subset A_j$, and $\mathcal{I^*}$ is an instance over only $A_i$, we have that $\MMS_j(\mathcal{I^*}) = \MMS_i(\mathcal{I})$. Thus,   $v_j(B''_j) \geq \MMS_i(\mathcal{I})$. 
    
    We can then transform Equation~\ref{eq:lam2} as follows: $$v_j(B_j) > v_j(B'_j) - \MMS_i(\mathcal{I}) + \MMS_i(\mathcal{I}),$$
    meaning it must be the case that $v_j(B_j) > v_j(B'_j)$. Because we know agent $j$ has identical valuations in $\mathcal{I}$ and $\mathcal{I'}$, and $\alloc'$ is an MMS allocation, we can conclude that agent~$j$ receives at least their maximin fair share in $\alloc$. 
\end{itemize}
\noindent Thus we have shown for any agent $j\in N$ that they receive their maximin fair share in the allocation~$\alloc$, meaning it must be an MMS allocation. Since~$\mathcal{I}$ was an arbitrary instance where exactly $k$ agents submit the approval set $M$, this concludes the proof. \qed
\end{proof}

We can now (finally) present the main result of this section. 

\begin{theorem}\label{thm:laminar}
For $n$ agents with cost utilities and laminar set approvals, there always exists an MMS allocation. 
\end{theorem}

\begin{proof}
First, note that given agents with laminar set approvals, if no agent has $\goods$ as their approval set, then we can find a $k$-partition $(N_1, \dots, N_k)$ of agents and pairwise disjoint subsets $\goods_1, \dots, \goods_k$ of items such that agents in $N_\ell$ do not approve any items in $\goods \setminus \goods_\ell$, and there is an agent $i \in N_\ell$ such that $A_i = M_\ell$. It is clear---because the agents are partitioned such that each partition considers a distinct set of items from $M$---that if we find an MMS allocation for each of the $k$ sub-cases, this gives us an MMS allocation in the global case. Therefore, without loss of generality, we assume for any instance that at least one agent submits $\goods$ as their approval set. 

%observe that in any instance where agents with cost utilities have laminar set approvals, there is at least one agent whose approval set is $M$. \todo{not necessarily, need to make this more precise} 
Now suppose there are $n$ agents with cost utilities who all submit $\goods$ as their approval set. Then an MMS allocation trivially exists. Applying Lemma~\ref{lem:lam-ind} inductively, we see that for agents with cost utilities and laminar set approvals, an MMS allocation always exists given that at least one agent submits $M$ as their approval set.\qed \end{proof}

\begin{remark}
If an MMS allocation exists, then an MMS and PO allocation always exists since after each Pareto improvement agent's utility is weakly increasing. Thus, Theorem~\ref{thm:laminar} implies that under cost utilities and laminar set approvals, MMS and PO allocation always exist.
\end{remark}

\section{Strategyproof MMS Allocations}\label{sec:strat}
In this section, we study the strategic guarantees possible under cost utilities. We first show that for cost utilities, the \emph{Sequential Allocation} mechanism is strategyproof. Let us first define what we mean by strategyproofness. 

An allocation mechanism~$f$ is \emph{manipulable} if there is some agent~$i \in N$ such that $v_i(f(\boldsymbol{v}_{-i}, v'_i)_i) > v_i(f(\boldsymbol{v})_i)$, where $(\boldsymbol{v}_{-i}, v'_i)$ is the valuation that results when $v_i$ is replaced by $v'_i$. In other words, agent~$i$ can misrepresent their preferences by submitting an untruthful valuation~$v'_i$, thereby getting a more preferred outcome. We say $f$ is \emph{strategyproof} if it is not manipulable by any agent. 

We now define the \emph{Sequential Allocation mechanism} from previous studies \citep{knwxaaai13,knwijcai13}. We first define a \emph{picking sequence} as a sequence of agents in $N$. Note that the sequence of agents can be of any length, and any agent might appear multiple times in the sequence. We can think of Sequential Allocation as proceeding sequentially (as the name indicates), through the ordering of agents. At each step, the agent whose turn it is chooses the item with the highest cost that a) is still available and b) is in their approval set. Note that we ``force'' agents to pick their most wanted item, as reported in their approvals. If there are no remaining items that an agent finds useful then we skip this agent and continue with the next. The mechanism allows some items to remain unallocated only if they are not approved by any agent. 

In fact, Sequential Allocation is a family of mechanisms, each defined by the picking sequence. As we will see, the properties of the mechanism also heavily depend on the picking sequence in question. For example, it is well known that Sequential Allocation is not strategyproof in general unless an agent’s picks are all consecutive \citep{knwxaaai13}.

In the rest of this section, we will assume that the goods in $M = \{g_1, \dots, g_m\}$ are ordered from lowest cost to highest cost---i.e. $c(g_k) \leq c(g_\ell)$ for all $k < \ell$.

\begin{proposition}\label{prop:strat}
For agents with cost utilities, there exists a picking sequence such that Sequential Allocation is strategyproof and results in a Pareto efficient allocation.\footnote{We prove Proposition~\ref{prop:strat} for a picking sequence used in the proof of Proposition~\ref{prop:poss}, but note that there are simpler picking sequences for which it holds.} 
\end{proposition}
\begin{proof}

We define a sequence $S$ of agents of length $n+2$, and a sequence $T$ of agents where every agent appears exactly once. Let $S = 1, 2, \dots, n-1, n, n, n$, and $T = n, n-1, ..., 2, 1$. Our picking sequence is $S$, followed by $m$ copies of each element in the sequence $T$. We can think of this as running through $S$, then letting each agent in $T$ choose all the items they want when it is their turn in $T$. We now show that this gives us a strategyproof mechanism. 

It is immediately clear that agent~$n$ has no incentive to manipulate. They cannot move themselves up in the picking sequence, and once it is their turn, they can essentially grab all the items they want. 

For any other agent~$i \in N$, let $X_i$ be the items remaining immediately before agent $i$ received their first item, and let $x$ be the item with highest cost in $X_i \cap A_i$. Then, agent $i$ receives $x$. After this, all items in the approval sets of agents $i+1, i+2, ..., n$ are allocated before agent $i$ receives all remaining items in $A_i$. Thus, agent $i$ receives the bundle $x \cup (A_i \cap (X_i \setminus (A_{i+1} \cup A_{i+2} \cup ... \cup A_{n} )))$. Note that the preferences of agent $i$ do not decide the set $X_i$. Hence, by misreporting, agent $i$ is unable to gain any additional items that they approve. 

Note that the final allocation is Pareto optimal because items are only allocated to agents that want them. As agents have cost utilities, all agents who want an item will value it the same. This concludes the proof.\qed \end{proof}

We now consider whether there are picking sequences that can give us an MMS allocation along with truthfulness for a restricted number of items. Such a restriction is needed because computation of an agent's MMS value is NP-hard for an arbitrary number of items, which implies that no picking sequence is guaranteed to output an MMS allocation.\footnote{This is under the assumption  P$\neq$NP.} We start with a lemma that will be used to prove Proposition~\ref{prop:poss}.

\begin{lemma}\label{lem:strat}
For $n$ agents and $n+2$ goods, let $|A_i|\geq n+k$ where $k \in \{0,1,2\}$. The $(n-k)$-th most valuable item in $A_i$ is guaranteed to give agent~$i$ their maximin fair share. 
\end{lemma}
\begin{proof}
Note that for any $n$-partition of the items in $A_i$, there is at most $k$ bundles that are not singletons, meaning at least $(n-k)$ of the bundles have just a single item. Any of these bundles will give agent~$i$ their maximin fair share. Of these $(n-k)$ singleton bundles, the highest possible value for the lowest valued bundle is the cost of the $(n-k)$-th most valuable item in the agent's approval set. 
\end{proof}

\begin{restatable}{proposition}{possi}\label{prop:poss}
For $n$ agents with cost utilities, and $n+2$ goods, there exists a picking sequence such that Sequential Allocation is strategyproof, and returns a Pareto efficient MMS allocation. 
\end{restatable} 

\begin{proof}
We first show that there is a picking sequence such that Sequential Allocation returns an MMS allocation. If an agent approves fewer than $n$ items, they still receive their maximin fair share even when no items are allocated to them. We therefore focus on agents who approve at least $n$ items. We define the picking sequence based on the cost of the items in $M$.  

\begin{itemize}
\item[$\blacktriangleright$] If $c(g_4) > c(g_2) + c(g_3)$, our picking sequence is $1, 2, \dots, n-1, n, n, n$. 
\item[$\blacktriangleright$] Otherwise, our picking sequence is $1, 2, \dots, n-1, n, n, n-1$. 
\end{itemize}

\noindent Note that these differ only in who gets to pick the last item. The fact that agents~$1$ through $n-2$ are guaranteed their maximin fair share for both picking sequences follows from Lemma~\ref{lem:strat}. It remains to show that the same holds for agent~$n-1$ and agent~$n$. If agent~$n-1$ or agent~$n$ approve at most $n$ items, then we already know they are guaranteed their maximin fair share. If agent~$n-1$ approves $n+1$ items, their $(n-1)$-th most valuable item is still up for grabs, and by Lemma~\ref{lem:strat} this will guarantee them their maximin fair share. 

We now consider what happens when agent~$n$ approves $n+k$ items---for $k \in \{1,2\}$), and when agent $(n-1)$ approves $n+2$ items. We look at each potential picking sequence separately.

\textit{Case 1:} Suppose $c(g_4) > c(g_2) + c(g_3)$.  If agent~$n$ approves $n+k$ items, they will receive at least $k+1$ items, as they pick last and can pick up to three items if they want, given the picking sequence $1, 2, \dots, n-1, n, n, n$. Clearly a bundle of size $k+1$ guarantees them their maximin fair share. 

What remains is to check what happens when agent $(n-1)$ approves all items in $M$, so suppose this to be the case. 
We first show that the maximin fair share of agent~$n-1$ is $\min(c(\{g_1, g_2, g_3\}), c(g_4))$.
Consider a partition $\alloc$ of $M$ into $n$ bundles, where $c(B_i) \geq \MMS^{n}_{n-1}(\goods)$ for each $i \in N$. 
At least $n-2$ of these bundles must contain a single item, and so we know that either \textit{i)} $n-2$ bundles contain one item and two bundles contain two items, or \textit{ii)} $n-1$ bundles contain one item and one bundle contains three items. We know that $c(g_4) > c(g_2) + c(g_3)$ by assumption, and the non-singleton bundles will be made up of the four lowest value items---$g_1, \dots, g_4$. Then the best we can do is one $3$-item bundle $B = \{g_1, g_2, g_3\}$ and all the remaining items in singleton bundles. It follows that the maximin fair share of agent~$n-1$ is $\min(c(B), c(g_4))$. When it is agent~$n-1$'s turn to pick, in the worst case, the only remaining goods will be $g_1, \dots, g_4$, in which case agent~$n-1$ can pick item $g_4$ to guarantee their maximin fair share.

\textit{Case 2:} Suppose instead that $c(g_4) \leq c(g_2) + c(g_3)$. If agent~$n$ approves $n+1$ items, they will receive two items, guaranteeing them their maximin fair share. If agent~$n$ approves all items in $M$, their maximin fair share in this case is determined by the lowest value bundle between the two bundles of size two, and the cheapest singleton. In particular, agent~$n$'s maximin fair share is $\min(c(\{c_1, c_4\}), c(\{c_2, c_3\}), c(\{g_5\}))$. With this picking sequence, agent~$n$ receives two items and in the worst case, this will be the bundle $B = \{g_2, g_3\}$. Clearly this guarantees agent~$n$ their maximin fair share. 

Finally, we look at when agent $(n-1)$ approves $n+2$ items. In this case, we know that their maximin fair share is determined by  $\min(c(\{g_1, g_4\}), c(\{g_2, g_3\}), c(\{g_5\}))$, as was the case for agent~$n$. As we did for agent~$n$ we know that agent~$(n-1)$ will receive two items, and in the worst case this will be the bundle $B = \{g_4, g_1\}$, which gives the agent their maximin fair share.

Strategyproofness and Pareto efficiency for the first case follows directly from Proposition~\ref{prop:strat}.
We now prove strategyproofness and Pareto efficiency for the second case, where $c(g_4) \leq c(g_2) + c(g_3)$. In this case, our picking sequence is $1, 2, \dots, n-1, n, n, n-1$. 

For any agent $i \in N$, if $i < n-1$, it is clear that there is no way for the agent to manipulate as they only get one pick. For agent~$n$, because their picks are right after each other, they also have no incentive to manipulate. Thus, we need only consider agent~$n-1$. 
Let $X$ be the items remaining immediately before agent $n-1$ received their first item, and let $x$ be the item with highest cost in $X \cap A_{n-1}$. Agent $n-1$ will pick $x$ by definition of the mechanism. Agent $n$ then receives their two highest valued remaining items if they exist (call these items $y$ and $y'$), and then finally agent~$n-1$ potentially receives the last item they approve (call this item $z$).

First, consider the case where agent~$n-1$ misreports that they approve some item $x'$, and they receive $x'$ instead of $x$. Then, the bundle of agent~$n-1$ will consist of $x'$ (which they value at 0), and potentially some other item $z'$ with $v_{n-1}(z') \leq v_{n-1}(x)$. Thus, agent~$n-1$ is not better off in this case.
Otherwise, if agent~$n-1$ instead misreports that they do not approve item $x$, then they will pick some other item $x''$ instead, where $c(x'') \leq c(x)$. If $x'' \neq y$ and $x'' \neq y'$, then we must have $v_{n-1}(x'') \leq v_{n-1}(z)$, and so agent~$n-1$ is not better off. Otherwise, if $x'' = y$ or $x'' = y'$, then agent~$n-1$ will have strictly fewer options for their final pick (compared to the case where they do not misreport), and so they are still not any better off.

Therefore, the mechanism is strategyproof. It is clear that no agent is assigned an item they do not want, and all items that are wanted by at least one agent are assigned to someone. Thus the allocation is Pareto efficient. \qed 
\end{proof}

We remark that Proposition~\ref{prop:poss} is tight in the sense that it no longer holds when there are $n$ agents and $n+3$ items.

\begin{restatable}{proposition}{last}\label{prop:counter} For agents with cost utilities, there exists an instance with $n=2$ agents and $m=5$ goods such that no strategyproof mechanism can guarantee a Pareto efficient MMS allocation. \end{restatable}

\begin{proof}
Let $n = 2$, and $M = \{g_2, g_3, g_4, g_5, g_6\}$ such that $c(g_i) = i$. We will show that no allocation mechanism can satisfy strategyproofness while also guaranteeing a Pareto Efficient MMS allocation. Our aim is to start from an instance~$\mathcal{I}_1$ and---by repeatedly applying the three axioms---reach a contradiction.

First, consider the instance~$\mathcal{I}_1$, where both agents approve all items---this corresponds to the top row of Table~\ref{tab:mms}. Then, their maximin fair share is $10$, and the only way to reach an MMS allocation is to give $g_4$ and $g_6$ to one agent, and $g_2, g_3$ and $g_5$ to another. Suppose without loss of generality that $\{g_2, g_3, g_5\}$ is allocated to agent~$1$, and $\{g_4, g_6\}$ is allocated to agent~$2$. We will consider $5$ further instances. 

$\mathcal{I}_2$ differs only on agent~$2$'s approval set---they now only approve items $g_4$, $g_5$, and $g_6$. By strategyproofness, agent~$2$ must still receive a bundle she values at $10$. If this were a higher value the agent could manipulate from $\mathcal{I}_1$, and if it were lower, they could manipulate from $\mathcal{I}_2$ to $\mathcal{I}_1$. 

Instance~$\mathcal{I}_3$  differs from instance~$\mathcal{I}_2$ only on agent~$1$'s approval set---they now only approve items $g_3$, $g_4$, $g_5$, and $g_6$. As agent~$1$ is the only one approving item $g_3$, they must be allocated this item by Pareto efficiency. The maximin value of agent~$1$ in this instance is $9$, so they must receive one of the following bundles:  $\{(g_3, g_6\}$, $\{g_3, g_4, g_5\}$, $\{g_3, g_4, g_6\}$, $\{g_3, g_5, g_6\}$, or $\{g_3, g_4, g_5, g_6\}$. All but $\{g_3, g_6\}$ break strategyproofness, as agent~$1$ would have an incentive to manipulate from $\mathcal{I}_2$ to $\mathcal{I}_3$. 

Instance~$\mathcal{I}_4$ differs from instance~$\mathcal{I}_3$ only on agent~$1$'s approval set---they now only approve item $g_6$. Agent~$1$ must be allocated $g_6$. If this were not the case, they would have an incentive to manipulate from $\mathcal{I}_4$ to $\mathcal{I}_3$ as they do receive item $6$ in that instance. 

Instance~$\mathcal{I}_5$ differs from instance~$\mathcal{I}_4$ only on agent~$1$'s approval set---they now approve items $g_2$, $g_3$ and $g_6$.  As agent~$1$ is the only one approving items $g_2$ and $g_3$, they must be allocated these items by Pareto efficiency. If agent~$1$ is not also given $g_6$, they would have an incentive to manipulate from $\mathcal{I}_4$ to $\mathcal{I}_3$ as their bundle in that instance is valued at $6$ (which is greater than $2+3$, the value of the bundle $\{g_2, g_3\}$). Note that this gives them a bundle valued at $11$. 

Finally, instance~$\mathcal{I}_6$ differs from instance~$\mathcal{I}_5$ only on agent~$1$'s approval set---they now approve all items.  If agent~$1$ is given a bundle valued lower than $11$, they would have an incentive to manipulate from $\mathcal{I}_6$ to $\mathcal{I}_5$. Note however that $\mathcal{I}_6 = \mathcal{I}_2$, and our axioms dictated in that instance that agent~$1$ must receive utility of $10$. This gives us our contradiction. 
\begin{table}
\begin{center}
\begin{tabular}{ c | c | c  }
\textbf{Instance} & \textbf{Approval Sets} & \textbf{Allocation}  \\ 
 \hline
 $\mathcal{I}_1$ & $(23456) (23456)$ &  $(235) (46)$   \\  
 $\mathcal{I}_2$ & $(23456) (456)$ & $(235) (46)$     \\
 $\mathcal{I}_3$ & $ (3456) (456)$ & $(36) (45)$  \\
 $\mathcal{I}_4$ & $(6) (456)$ & $(6) (45)$ \\
 $\mathcal{I}_5$ & $ (236) (456)$ & $ (236) (45)$  \\
 $\mathcal{I}_6$ & $ (23456) (456) $ & (at least 11) (at most 9)
\end{tabular}
\end{center}
\caption{Table showing the approval sets corresponding to each instance in the proof of Proposition~\ref{prop:counter}. For example, $(23456) (456)$ denotes the instance where agent $1$ approves all items, and agent $2$ approves items $g_4$, $g_5$, and $g_6$.   The second column describes outcomes consistent with MMS, Pareto efficiency, and strategyproofness. Note that we omit items not approved by either agents, as they can be allocated to anyone without affecting any of the three axioms.}
\label{tab:mms}
\end{table}
\qed 
\end{proof}
\section{Conclusion}\label{sec:concl}

Fair division of indivisible resources is a challenging yet important problem with wide-ranging applications. In this paper, we have established that cost utilities are a useful restriction to study, especially in the context of MMS allocations. We have shown that there are several classes of instances where MMS allocations always exist under cost utilities. We also show that cost utilities are helpful in circumventing problems of strategic manipulation. 

The topic of MMS allocations in general, and for cost utilities in particular, poses many challenging questions. One might consider various fair division problems with constraints under cost utilities. A prime example is cardinality constraints---or more generally, budget constraints---which are quite natural in this setting. 

Our work serves as a further indication that fair division under cost utilities is a fruitful research direction. 

\subsubsection*{Acknowledgements.}  

This project was partially supported by the ARC Laureate Project FL200100204 on ``Trustworthy AI''.

\renewcommand\bibname{References}
\bibliographystyle{splncs04nat}
\bibliography{costutils}

\begin{thebibliography}{27}
\providecommand{\natexlab}[1]{#1}
\providecommand{\url}[1]{\texttt{#1}}
\providecommand{\urlprefix}{URL }
\expandafter\ifx\csname urlstyle\endcsname\relax
  \providecommand{\doi}[1]{doi:\discretionary{}{}{}#1}\else
  \providecommand{\doi}{doi:\discretionary{}{}{}\begingroup
  \urlstyle{rm}\Url}\fi

\bibitem[{Akrami et~al.(2022)Akrami, Rezvan, and Seddighin}]{ARS22}
Akrami, H., Rezvan, R., Seddighin, M.: An {EF2X} allocation protocol for
  restricted additive valuations. In: Proceedings of the 31st International
  Joint Conference on Artificial Intelligence (IJCAI) (2022)

\bibitem[{Amanatidis et~al.(2017{\natexlab{a}})Amanatidis, Birmpas,
  Christodoulou, and Markakis}]{ABCM17}
Amanatidis, G., Birmpas, G., Christodoulou, G., Markakis, E.: Truthful
  allocation mechanisms without payments: Characterization and implications on
  fairness. In: Proceedings of the 18th {ACM} Conference on Economics and
  Computation (EC) (2017{\natexlab{a}})

\bibitem[{Amanatidis et~al.(2022)Amanatidis, Birmpas, Filos{-}Ratsikas, and
  Voudouris}]{ABFV22}
Amanatidis, G., Birmpas, G., Filos{-}Ratsikas, A., Voudouris, A.A.: Fair
  division of indivisible goods: {A} survey. In: Proceedings of the 31st
  International Joint Conference on Artificial Intelligence (IJCAI) (2022)

\bibitem[{Amanatidis et~al.(2017{\natexlab{b}})Amanatidis, Markakis, Nikzad,
  and Saberi}]{AMNS17}
Amanatidis, G., Markakis, E., Nikzad, A., Saberi, A.: Approximation algorithms
  for computing maximin share allocations. ACM Transactions on Algorithms
  \textbf{13}(4), 1--28 (2017{\natexlab{b}})

\bibitem[{Asadpour et~al.(2012)Asadpour, Feige, and Saberi}]{AFS12}
Asadpour, A., Feige, U., Saberi, A.: {S}anta {C}laus meets hypergraph
  matchings. ACM Transactions on Algorithms \textbf{8}(3) (2012)

\bibitem[{Babaioff et~al.(2021)Babaioff, Ezra, and Feige}]{BEF21}
Babaioff, M., Ezra, T., Feige, U.: Fair and truthful mechanisms for dichotomous
  valuations. In: Proceedings of the 35th {AAAI} Conference on Artificial
  Intelligence (AAAI) (2021)

\bibitem[{Bansal and Sviridenko(2006)}]{BS06}
Bansal, N., Sviridenko, M.: The {S}anta {C}laus problem. In: Proceedings of the
  38th Annual ACM Symposium on Theory of Computing, p. 31–40 (2006)

\bibitem[{Bouveret et~al.(2017)Bouveret, Cechl{\'a}rov{\'a}, Elkind, Igarashi,
  and Peters}]{bouveret2017fair}
Bouveret, S., Cechl{\'a}rov{\'a}, K., Elkind, E., Igarashi, A., Peters, D.:
  Fair division of a graph. In: Proceedings of the 26th International Joint
  Conference on Artificial Intelligence (IJCAI) (2017)

\bibitem[{Bouveret and Lema{\^\i}tre(2016)}]{BL16}
Bouveret, S., Lema{\^\i}tre, M.: Characterizing conflicts in fair division of
  indivisible goods using a scale of criteria. Autonomous Agents and
  Multi-Agent Systems \textbf{30}(2), 259--290 (2016)

\bibitem[{Brams and Taylor(1996)}]{BT96}
Brams, S.J., Taylor, A.D.: Fair Division: From cake-cutting to dispute
  resolution. Cambridge University Press (1996)

\bibitem[{Brandt et~al.(2016)Brandt, Conitzer, Endriss, Lang, and
  Procaccia}]{BCELP16}
Brandt, F., Conitzer, V., Endriss, U., Lang, J., Procaccia, A.D.: Handbook of
  Computational Social Choice. Cambridge University Press (2016)

\bibitem[{Budish(2011)}]{B11}
Budish, E.: The combinatorial assignment problem: Approximate competitive
  equilibrium from equal incomes. Journal of Political Economy \textbf{119}(6),
  1061--1103 (2011)

\bibitem[{Camacho et~al.(2022)Camacho, Fonseca-Delgado, {Pino Pérez}, and
  Tapia}]{CFPT22}
Camacho, F., Fonseca-Delgado, R., {Pino Pérez}, R., Tapia, G.: Generalized
  binary utility functions and fair allocations. Mathematical Social Sciences
  (2022)

\bibitem[{Cheng and Mao(2018)}]{CM18}
Cheng, S., Mao, Y.: Integrality gap of the configuration {LP} for the
  restricted max-min fair allocation  (2018),
  \urlprefix\url{http://arxiv.org/abs/1807.04152}

\bibitem[{Ebadian et~al.(2022)Ebadian, Peters, and Shah}]{EPS22}
Ebadian, S., Peters, D., Shah, N.: How to fairly allocate easy and difficult
  chores. In: Proceedings of the 21st International Conference on Autonomous
  Agents and Multiagent Systems (AAMAS) (2022)

\bibitem[{Feige et~al.(2021)Feige, Sapir, and Tauber}]{FST21}
Feige, U., Sapir, A., Tauber, L.: A tight negative example for {MMS} fair
  allocations. In: Proceedings of the 17th International Conference on Web and
  Internet Economics (WINE) (2021)

\bibitem[{Garg and Taki(2020)}]{GT20}
Garg, J., Taki, S.: An improved approximation algorithm for maximin shares. In:
  Proceedings of the 21st ACM Conference on Economics and Computation (EC)
  (2020)

\bibitem[{Ghodsi et~al.(2018)Ghodsi, Hajiaghayi, Seddighin, Seddighin, and
  Yami}]{GHSSY18}
Ghodsi, M., Hajiaghayi, M., Seddighin, M., Seddighin, S., Yami, H.: Fair
  allocation of indivisible goods: Improvements and generalizations. In:
  Proceedings of the 19th ACM Conference on Economics and Computation (EC)
  (2018)

\bibitem[{Goldman and Procaccia(2015)}]{GP15}
Goldman, J., Procaccia, A.D.: Spliddit: Unleashing fair division algorithms.
  SIGecom Exchanges \textbf{13}(2), 41--46 (2015)

\bibitem[{Halpern et~al.(2020)Halpern, Procaccia, Psomas, and Shah}]{HPPS20}
Halpern, D., Procaccia, A.D., Psomas, A., Shah, N.: Fair division with binary
  valuations: One rule to rule them all. In: Proceedings of the 16th
  International Conference on Web and Internet Economics (WINE), Springer
  (2020)

\bibitem[{Heinen et~al.(2018)Heinen, Nguyen, Nguyen, and Rothe}]{HNNR18}
Heinen, T., Nguyen, N., Nguyen, T.T., Rothe, J.: Approximation and complexity
  of the optimization and existence problems for maximin share, proportional
  share, and minimax share allocation of indivisible goods. Autonomous Agents
  and Multi-Agent Systems \textbf{32}(6), 741--778 (2018)

\bibitem[{Kalinowski et~al.(2013{\natexlab{a}})Kalinowski, Narodytska, and
  Walsh}]{knwijcai13}
Kalinowski, T., Narodytska, N., Walsh, T.: A social welfare optimal sequential
  allocation procedure. In: Proceedings of the 23rd International Joint
  Conference on Artificial Intelligence (IJCAI) (2013{\natexlab{a}})

\bibitem[{Kalinowski et~al.(2013{\natexlab{b}})Kalinowski, Narodytska, Walsh,
  and Xia}]{knwxaaai13}
Kalinowski, T., Narodytska, N., Walsh, T., Xia, L.: Strategic behavior when
  allocating indivisible goods sequentially. In: Proceedings of the 27th AAAI
  Conference on Artificial Intelligence (AAAI) (2013{\natexlab{b}})

\bibitem[{Kurokawa et~al.(2018)Kurokawa, Procaccia, and Wang}]{KPW18}
Kurokawa, D., Procaccia, A.D., Wang, J.: Fair enough: Guaranteeing approximate
  maximin shares. Journal of the ACM \textbf{65}(2), 1--27 (2018)

\bibitem[{Moulin(2004)}]{M04}
Moulin, H.: Fair division and collective welfare. MIT press (2004)

\bibitem[{Othman et~al.(2010)Othman, Sandholm, and Budish}]{ASB10}
Othman, A., Sandholm, T., Budish, E.: Finding approximate competitive
  equilibria: efficient and fair course allocation. In: Proceedings of the 9th
  International Conference on Autonomous Agents and Multiagent Systems (AAMAS)
  (2010)

\bibitem[{Vossen(2002)}]{V02}
Vossen, T.: Fair allocation concepts in air traffic management. PhD thesis
  (2002)

\end{thebibliography}

\end{document}